\documentclass[12pt]{article}

\usepackage{amsthm, amsmath,amssymb,amsfonts}
\usepackage{graphicx,color}
\usepackage[all]{xy} \input xy
\usepackage{natbib}
\usepackage{hyperref}
\hypersetup{colorlinks, citecolor=blue, filecolor=blue, linkcolor=blue, urlcolor=blue}

\newtheorem{theorem}{Theorem}[section]
\newtheorem{proposition}[theorem]{Proposition}
\newtheorem{corollary}[theorem]{Corollary}

\newtheorem{example}[theorem]{Example}
\theoremstyle{definition}
\newtheorem{definition}[theorem]{Definition}
\newtheorem{algorithm}[theorem]{Algorithm}

\def\A{\mathcal{A}}

\def\K{\mathbb{K}}
\def\F{\mathbb{F}}
\def\pj{\mathbb{P}}
\def\E{\mathbb{E}}
\def\C{\mathbb{C}}
\def\L{\mathbb{L}}
\def\N{\mathbb{N}}

\def\cP{\mathcal{P}}

\def\ss{\mathcal{S}}
\def\cL{\mathcal{L}}
\def\cH{\mathcal{H}}

\def\xx{\overline{x}}
\def\yy{\overline{y}}
\def\zz{\overline{z}}

\def\bb{\mathcal W}

\def\ccs{\mathcal{C}_{s}}

\def\ot{\,\overline{t}\,}

\def\res{\mathrm{Res}}
\def\pp{\mathrm{pp}}

\def\gg{\mathfrak{g}}

\begin{document}

\title{\vspace{-1em}First Steps Towards Radical Parametrization of Algebraic Surfaces}

\author{
J. Rafael Sendra\footnote{Supported by the Spanish \textsf{Ministerio de Ciencia e Innovaci\'on} under the Project MTM2008-04699-C03-01, and by the \textsf{Ministerio de Econom\'ia y Competitividad} under the project MTM2011-25816-C02-01; member of the of the Research Group \textsc{asynacs} (Ref. \textsc{ccee2011/r34}).} \\
Dep. de Matem\'aticas \\
Universidad de Alcal\'a \\
Alcal\'a de Henares, Madrid, Spain \\
\texttt{rafael.sendra@uah.es}
\and
David Sevilla\thanks{Supported by the Austrian Science Fund (FWF): P22766-N18.} \\
Johann Radon Institute (RICAM) \\
Altenbergerstrasse 69 \\
A-4040 Linz, Austria \\
\texttt{david.sevilla@oeaw.ac.at, \url{http://www.davidsevilla.com}} \\
}
\date{}
\maketitle

\begin{abstract}
We introduce the notion of radical parametrization of a surface, and we provide algorithms to compute such type of parametrizations for families of surfaces, like: Fermat surfaces, surfaces with a high multiplicity (at least the degree minus 4) singularity, all irreducible surfaces of degree at most 5, all irreducible singular surfaces of degree 6, and surfaces containing a pencil of low-genus curves. In addition, we prove that radical parametrizations are preserved under certain type of geometric constructions that include offset and conchoids.
\end{abstract}




\section{Introduction}
Let us try to motivate our work from two different points of view: from the purely mathematic point of view and from the potential applications.

Let $K$ be a field (say, of characteristic zero) and $f\in K[x]$ a univariate polynomial. A classical question is whether the roots of $f(x)$ can be computed exactly. We know by Galois theory that, in general, we can only provide a positive answer (based on radical expressions) when the degree is at most 4. Furthermore, methods to compute them are available. Now, take $f(x,y)\in K[x,y]$ and formulate the same question. The first remark is that the dimension of the zero-set (i.e. the set of roots) has increased from 0 to 1. The answer was given by Zariski, who proved in \cite{Zariski1926a} that this is only possible, in general, when the genus of the curve defined by $f(x,y)$ (say w.l.o.g. that $f$ is irreducible over the algebraic closure $\overline{K}$ of $K$) is at most 6. Furthermore, in \cite{SeSevilla}, it is shown how to calculate them if the genus is at most 4; the cases of genus 5 and 6 are treated algorithmically in \cite{Ha}. An expectable remark is that, when the genus is 0, the answer is expressed in $K'(t)$, where $K'$ is a finite algebraic extension of $K$ and $t$ is transcendental over $K'$, and when $1\leq$ genus $\leq 6$ the answer is given by radical expressions over $K'(t)$.

The next step is to formulate the same question when $f(x,y,z)\in K[x,y,z]$. This is the central topic of this article. To our knowledge, there exist no theoretical results establishing the limitations to solve the problem, either by the degree (in 1 variable) or by genus (in 2 variables). In this paper we give the first steps in this direction; later in this introduction we give more details. 

Now, let us see the problem from another side. It is well known that in many applications dealing with geometric objects, parametric representations are very useful. Examples of this affirmation are, for instance:
\begin{itemize}
 \item when dealing with the intersection of two varieties (say surfaces or curve/surface) it is convenient to have a parametric representation of one of them;
 \item when plotting in the screen a curve or a surface;
 \item when computing line or surface integrals;
 \item when dealing with the velocity or the acceleration of a particle following a path on the variety;
 \item when performing geometric transformations, such as rotations, translations, and scaling;
 \item when executing projections, etc.
\end{itemize}
Therefore, it is important to provide parametric representations. There are different options: one can use rational functions, trigonometric parametrizations, piecewise approximate parametrizations, etc. In this paper, we propose to enlarge the class of rational parametrizations by introducing radical parametrizations, i.e. algebraic expressions involving (maybe nested) radicals of polynomials.

To be more precise, we introduce the notion of radical parametrization by using basic notions of Galois theory as well as the ideas in \cite{SeSevilla}, and we provide algorithms to parametrize by radicals some families of surfaces. In Section \ref{sec-radicals}, besides introducing the basic notations, we show how to parametrize by radicals some special surfaces, including Fermat surfaces. In the next two sections we try to follow the first steps in the theoretical analysis of the rationality of surfaces (see \cite{Sch-a}). More precisely, in Section \ref{sec-parametrization-by-lines} we generalize the notion of parametrization by lines to the case of radicals. As a consequence we provide an algorithm that parametrizes every surface having a singularity of multiplicity $d-r$, where $d$ is the degree of the surface and $r\leq 4$. From these results we prove that every irreducible surface of degree at most 5 is parametrizable by radicals, and that every singular surface of degree 6 is also parametrizable by radicals.

In Section \ref{sec-pencil} we provide algorithms to parametrize by radicals surfaces with a pencil of genus $g$ curves, where $g\leq 4$, with some additional hypotheses when the genus is 1 or 4. Furthermore, we offer an alternative approach for surfaces with a pencil of non-hyperelliptic genus 4 curves, based on a known theoretical characterization of such curves; the detection of the hyperelliptic case and its reduction to other cases in this article is well known, see that section for further details. Finally, in Section \ref{sec-offsets} we prove that radical parametrizations are preserved under certain types of geometric constructions that include offseting and conchoids.

Throughout this paper we will use the following notation: $\F$ is an algebraically closed field of characteristic zero (say e.g. $\F=\C$), and $\ss\subset\F^3$ is an irreducible surface defined by the irreducible polynomial $F(x,y,z)\in \F[x,y,z]$.

The computations and images shown in what follows have been performed with Maple and Surfex respectively.

\section{The Notion of Radical Parametrization of a Surface}\label{sec-radicals}

We introduce the notion of radical parametrization of a surface by extending the notion of radical parametrization of a curve (see \cite{SeSevilla}). For this purpose, in the sequel,
\begin{itemize}
\item let $\ot=(t_1,t_2)$, where $t_1,t_2$ are transcendental elements over $\F$, and
\item let $\K=\F(\ot)$.
\end{itemize}
But first, we recall briefly the classical notion of
solvability by radicals. $f\in K[x]$ ($K$ is any field of characteristic zero) is \textsf{solvable by radicals over $K$} if there exists a finite tower of field extensions
\[ K=K_0\subset K_1\subset \cdots \subset K_r\]
such that
\begin{enumerate}
\item for $i=1,\ldots,r$, $K_i=K_{i-1}(\alpha_i)$ where $\alpha_i^{\ell_i}-c_i=0$ for some $\ell_i>0$ and $c_i\in K_{i-1}$;
\item the splitting field of $f$ over $K$ is contained in $K_r$.
\end{enumerate}
A tower as above is called a \textsf{root tower for $f$ over $K$}.

Intuitively speaking a radical parametrization of a surface is triple of rational algebraic expressions involving (possibly nested) radicals of polynomials such that its formal substitution, in the defining polynomial of the surface, yields zero; and such that its Jacobian has rank 2. A formal definition follows.

\begin{definition}\label{def-rad-param}
The surface $\ss$ is \textsf{radical or pa\-ra\-me\-tri\-za\-ble by radicals} if there exist
\begin{itemize}
 \item A field $\E$ which is the largest field in a root tower of some $f(x)\in \K[x]$ solvable over $\K$,
  \[
   \F(t_1,t_2)=\K=\K_0\subset \K_1\subset \cdots \subset \K_r=\E\,,
  \]
 \item $(R_1,R_2,R_3)\in \E^{3}$ satisfying
 \begin{itemize}
  \item[$\circ$] $F(R_1,R_2,R_3)=0$,
  \item[$\circ$] the rank of the jacobian of $(R_1,R_2,R_3)$ w.r.t. $\ot$ is 2.
 \end{itemize}
\end{itemize}
In that case we call $(R_1,R_2,R_3)$ a \textsf{radical (affine) parametrization of $\ss$}. Similarly we introduce the notion of \textsf{radical projective parametrization}. We will denote the radical parametrization as $\cP(\ot)$. Furthermore, if $\L$ is a subfield of $\F$, and the root tower can be constructed over $\L(\ot)$, we say that $\ss$ is \textsf{parametrizable by radicals over $\L$}.
\end{definition}

\begin{example}
Let us start with a simple example to illustrate the notion of rational parametrization. We see that
\[
 \cP(\ot)=\left(t_1,t_2,\sqrt{1-t_{1}^{2}-t_{2}^{2}}\right)
\]
is a radical parametrization of the unit sphere. Indeed, we take $f(x)=x^2-(t_1^2+t_2^2)\in \K[x]$ (note the ambiguity of the $\sqrt{\ }$ symbol, since it could denote any of the two roots of $f$). Then $f$ solvable using the tower
\[
 \F(\ot)=\K_0\subset \K_1=\K_0\left(\sqrt{1-t_{1}^{2}-t_{2}^{2}}\right)=\E.
\]
Now, $\cP(\ot)\in \E^3$ and its jacobian has rank 2.
\end{example}

If $R_1,R_2,R_3\in\K$ then they trivially belong to the largest field of a root tower. Therefore, every rational surface parametrization is a radical parametrization, in the sense of our definition. In other words, every rational surface is radical. Also, note that if
\[
 \Phi:\ss \subset \F^3\rightarrow \Phi(\ss)\subset \F^3
\]
is a rational map of finite degree, and $\cP(\ot)$ is a radical parametrization of $\ss$, then $\Phi(\cP(\ot))$ is a radical parametrization of the Zariski closure of $\Phi(\ss)$; furthermore, because of the condition of the finite degree, the image variety is also a surface. As consequence, we get the following proposition.

\begin{proposition}
The property of being parametrizable by radicals is invariant under birational transformations.
\end{proposition}

Let us see some families of surfaces that can be easily parametrized by radicals. The basic idea in all cases below is to achieve a degree 4 polynomial from the implicit equation that we solve by radicals. We recall that $\ss$ is assumed to be irreducible. Note that, in all cases below, the surfaces are parametrized over the field of definition of $\ss$, i.e. over the smallest field where the defining polynomial $F$ can be expressed.

\paragraph{Case 1} If the partial degree of $F$ w.r.t. one of the variables is less or equal to 4, then $\ss$ is parametrizable by radicals. Let
\[
 F(x,y,z)=f_4(x,y) z^4 + f_3(x,y) z^3 + f_2(x,y) z^2 + f_1(x,y) z + f_0(x,y),
\]
then $g(z)=F(t_1,t_2,z)\in \K[z]$ has degree at most 4. Let $R(\ot)$ be a root of $g(z)$ (note that $g$ is solvable by radicals over $\K$) then
\[
 (t_1,t_2,R(\ot))
 \]
is a rational parametrization of $\ss$.

\paragraph{Case 2}An extension of the above situation is as follows. Let
\[
 F(x,y,z)= F_{m_1}(x,y) z^{sm_1+4} + F_{m_2}(x,y) z^{sm_2+3} + F_{m_3}(x,y) z^{sm_3+2}
\]
\[
 +F_{m_4}(x,y) z^{sm_4+1} + F_{m_5}(x,y)
\]
where $s,m_i\in\N$, and $F_i$ is homogeneous of degree $m_i$. We consider the rational transformation
\[\begin{array}{rccc}
 \Phi: & \Sigma:=\ss \setminus \{(a,b,c)\in \ss\,|\, c\neq 0\} &\rightarrow & \Phi(\Sigma)\subset \K^3 \\
   & (x,y,z) & \mapsto & (xz^s,yz^s,z)
\end{array}\]
Note that $\Phi$ is birational, indeed
\[\begin{array}{rccc}
 \Phi^{-1}: & \Phi(\Sigma) &\rightarrow & \Sigma \subset \K^3 \\
   & (x,y,z) & \mapsto & \left(\displaystyle\frac{x}{z^s},\frac{y}{z^s},z \right)
\end{array}\]
Then, we consider the surface $\tilde{\ss}$ defined as the Zariski closure of $\Phi(\Sigma)$. It is defined by
\[
 \tilde{F}(x,y,z) = F\left(\frac{x}{z^s},\frac{y}{z^s},z\right) = F_{m_1}(x,y) z^4 + F_{m_2}(x,y) z^3 + F_{m_3}(x,y) z^2
\]
\[
 + F_{m_4}(x,y) z + F_{m_5}(x,y)
\]
that corresponds to our first case. So, it is parametrizable by radicals. Let $(R_1(\ot),R_2(\ot),R_3(\ot))$ be a radical parametrization of $\tilde{\ss}$. Then,
\[
 \left(\frac{R_1(\ot)}{R_3(\ot)^s}, \frac{R_2(\ot)}{R_3(\ot)^s}, R_3(\ot)\right)
\]
parametrize $\ss$ radically.

\paragraph{Case 3} Let
\[
 F(x,y,z)=f(x,y) z^m-g(x,y).
\]
then $\ss$ is parametrizable by
\[
 \left(t_1,t_2, \sqrt[m]{\frac{g(t_1,t_2)}{f(t_1,t_2)}} \right).
\]
Note that if $f=1$ and $g(x,y)=-(x^m+y^m+1)$ we get the Fermat surfaces.

\paragraph{Case 4} A natural extension of the previous case is as follows. Let
\[
 F(x,y,z)= f_4(x,y) z^{4m} + f_3(x,y) z^{3m} + f_2(x,y) z^{2m} + f_1(x,y) z^m + f_0(x,y).
\]
If $(R_1(\ot),R_2(\ot),R_3(\ot))$ is a radical parametrization of the surface defined by
\[
 f_4(x,y) z^4 + f_3(x,y) z^3 + f_2(x,y) z^2 + f_1(x,y) z + f_0(x,y)
\]
then
\[
 \left({R_1(\ot)},{R_2(\ot)},\sqrt[m]{R_3(\ot)}\right)
\]
is a radical parametrization of $\ss$.

\section{Radical Parametrization by Lines}\label{sec-parametrization-by-lines}

In this section we will see how the idea of rationally parametrizing by lines can be extended to the case of radical parametrizations. In this section, we assume that $\ss$ has degree $n$ and $P\in \ss$ is an $(n-r)$-fold point, where $r\leq 4$. We consider a plane $\Pi$ such that $P\not\in \Pi$. Then, the projection from $P$ of $\ss$ over $\Pi$ is $r:1$. Therefore, using that univariate polynomials of degree at most $4$ are soluble by radicals we can generate a radical parametrization of $\ss$.

More precisely, the algorithmic reasoning is as follows: let $\cH(t_1,t_2)$ be a polynomial parametrization of the plane $\Pi$. We consider the line $\cL$ passing through $P$ and a generic point of $\Pi$. $\cL$ can be defined by $L(h)=P+h(\cH(t_1,t_2)-P)$ where $h$ is a new parameter. In this situation, we compute the interesection ${\cal L}\cap \ss$. The polynomial $f(h)=F(L(h))\in \K(h)$ has degree $n$ and factors as
\[
 f(h)=h^{n-r} g(h)
\]
so $\deg(h(g))=r\leq 4$. Therefore, it is solvable by radicals. This means that we can express $h$ as a radical expression on $\ot$, say $h=R(\ot)$, from where one concludes that $L(R(\ot))$ is a radical parametrization of $\ss$. Thus, we have the following theorem:

\begin{theorem}\label{theorem-param-by-lines}
Every irreducible surface of degree $n$ having an $(n-r)$-fold point is parametrizable by radicals.
\end{theorem}

In addition, note that the following algorithm can be derived.

\begin{algorithm}[Radical Parametrization by Lines]
Given an irreducible surface $\ss$ defined by $F$ and a point $P=(a,b,c)\in \ss$ of multiplicity $\deg(F)-r$ such that $r\leq 4$, the algorithm generates a radical parametrization of $\ss$.
\begin{enumerate}
 \item Let $L(h)=P+h(t_1-a,t_2-b,\lambda-c)$, where $\lambda\neq c$.
 \item Compute, by radicals, the roots of $g(h)=\displaystyle\frac{L(h)}{h^{\deg(F)-r}}$; say $R(\ot)$ is one of the roots.
 \item Return $\cP(\ot)={\cal H}(R(\ot),t_1,t_2))$.
\end{enumerate}
\end{algorithm}

We illustrate the algorithm with an example.

\begin{example}\label{ex-2}
We consider the surface $\ss$ (see Fig. \ref{fig-1}) defined by
\[
 F(x,y,z)={x}^{10}+{y}^{10}+{z}^{10}-xyz^4.
\]
It is a degree 10 surface with a 6-fold point at the origin. Therefore, the algorithm is applicable.
\begin{figure}
  \centering
    \includegraphics[width=5cm]{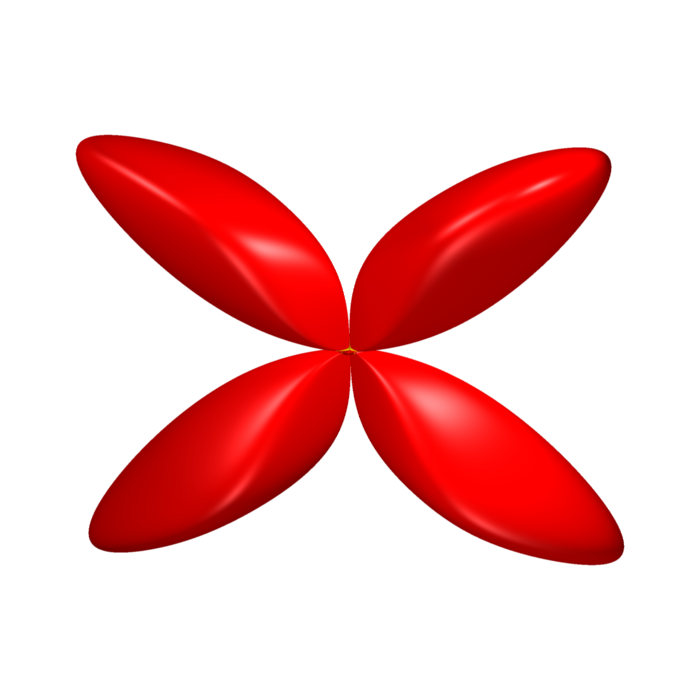}
  \caption{Surface in Example \ref{ex-2}}
  \label{fig-1}
\end{figure}
In Step 1 we get
\[
 L(h)=(h t_1,h t_2,h)
\]
In Step 2, the polynomial $g(h)$ is
\[
 g(h)=h^4{t_1}^{10}+h^4{t_2}^{10}+h^4-t_1t_2.
\]
Computing the roots one may take
\[
 R(h)={\frac {\sqrt { \left( {t_1}^{10}+{t_2}^{10}+1 \right) \sqrt { \left( {t_1}^{10}+{t_2}^{10}+1 \right) t_1t_2}}}{{t_1}^{10}+{t_2}^{10}+1}}.
\]
Thus, the algorithm generates the radical parametrization
\[
 \cP(\ot)=\left({\frac {\sqrt { \left( {t_1}^{10}+{t_2}^{10}+1 \right) \sqrt { \left( {t_1}^{10}+{t_2}^{10}+1 \right) t_1t_2}}\,\,t_1}{{t_1}^{10}+{t_2}^{10}+1}}, \right.
\]
\[
 \left. {\frac {\sqrt { \left( {t_1}^{10}+{t_2}^{10}+1 \right) \sqrt { \left( {t_1}^{10}+{t_2}^{10}+1 \right) t_1t_2}}\,\,t_2}{{t_1}^{10}+{t_2}^{10}+1}}, \right.
\]
\[
 \left. {\frac {\sqrt { \left( {t_1}^{10}+{t_2}^{10}+1 \right) \sqrt { \left( {t_1}^{10}+{t_2}^{10}+1 \right) t_1t_2}}}{{t_1}^{10}+{t_2}^{10}+1}} \right).
\]
\end{example}

From Theorem \ref{theorem-param-by-lines}, one deduces the following corollaries:

\begin{corollary}\label{cor-1}
Every irreducible surface of degree less or equal 5 is parametrizable by radicals.
\end{corollary}
\begin{proof}
Take a point of the surface and apply Theorem \ref{theorem-param-by-lines}.
\end{proof}

\begin{corollary}\label{cor-2}
Every singular irreducible surface of degree less or equal 6 is parametrizable by radicals.
\end{corollary}
\begin{proof}
Take a singular point of the surface and apply Theorem \ref{theorem-param-by-lines}.
\end{proof}

\section{Radical Parametrization of Surfaces with a Pencil of Low Genus Curves}\label{sec-pencil}

In this section we consider the case when $\ss$ has a pencil of curves $\ccs$ with genus less or equal 4. If the genus $\gg$ of the curves in the pencil is zero, it is known (see e.g. \cite{N70}, \cite{P97}, \cite{Sch-a}, \cite{Sch-b}) that $\ss$ is rational. We analyze the situation when $1\leq \gg \leq 4$, and we will be able to prove (providing and algorithm) that for $\gg\in \{2,3\}$ the surface $\ss$ is radical. Moreover, for $\gg\in\{1,4\}$, with some additional hypotheses, we also prove that the surface is radical, and in the case $\gg=4$ we offer an alternative method based on the fact that such curves are known to be trigonal (that is, they admit a $3:1$ map to the line).

We see $\ccs$ as a space curve in $\F(s)^3$. Then we can consider a projection $\pi:\F(s)^3\rightarrow \F(s)^2$ mapping $\ccs$ birationally onto a plane curve. Thus, we may assume w.l.o.g. that $\ccs$ is indeed a pencil of plane curves. Let us also assume w.l.o.g. that its defining polynomial is $F(x,y,s)$.

In the following reasoning we will apply to $\ccs$ the radical parametrization algorithms for curves given in \cite{SeSevilla}. So, for the sake of completeness, we recall here the main steps of the two main algorithms.

\begin{algorithm}[Radical parametrization of curves of genus $\leq 3$]\label{alg curves g<=3}
Let $G(x,y,w)$ be the homogeneous form defining a projective plane curve of genus $\gg \leq 3$ and degree $d$, then the algorithm computes a radical parametrization of the affine curve defined by $G(x,y,1)$ in terms of the variable $t$.
 \begin{enumerate}
  \item Compute the linear system of adjoints $\A_{d-2}$ of degree $d-2$.
  \item Take $(d-3)+\gg$ simple points on the curve.
  \item Compute the linear subsystem $\A_{d-2}^*$ of $\A_{d-2}$ obtained by forcing $\A_{d-2}$ to pass through all the simple points computed in Step 2. Let $\A_{d-2}^*$ be defined by the polynomial $H^*(x,y,w,t)$.
  \item Determine
  \[
   S_1(x,w,t)=\pp_{t}(\res_{y}(G,H^*)), \quad S_2(y,w,t)=\pp_{t}(\res_{x}(G,H^*))
  \]
  where $\pp_{t}$ denotes the primitive part with respect to $t$.
  \item Solve by radicals $S_{1}(x,1,t)$ and $S_{2}(y,1,t)$. Combining the roots of both polynomials (see \cite{SeSevilla} for further details on how to combine the roots) one gets the radical parametrization.
 \end{enumerate}
\end{algorithm}

\begin{algorithm}[Radical parametrization of curves of 2 $\leq$ genus $\leq 4$]\label{alg curves 2<=g<=4}
Let $G(x,y,w)$ be the homogeneous form defining a projective plane curve of genus $2\leq \gg \leq 4$ and degree $d$, then the algorithm computes a radical parametrization of the affine curve defined by $G(x,y,1)$ in terms of the variable $t$.
 \begin{enumerate}
  \item Compute the linear system of adjoints $\A_{d-3}$ of degree $d-3$.
  \item Take $\gg-2$ simple points on the curve.
  \item Compute the linear subsystem $\A_{d-3}^*$ of $\A_{d-3}$ obtained by forcing $\A_{d-3}$ to pass through all the simple points computed in Step 2. Let $H^*(x,y,w,t)$ be the defining polynomial of $\A_{d-3}^*$.
  \item Follow Steps 4 and 5 in Algorithm 1.
 \end{enumerate}
\end{algorithm}

The idea now is to apply either Algorithm \ref{alg curves g<=3} or Algorithm \ref{alg curves 2<=g<=4} to $\ccs$ as a curve over the algebraic closure of $\F(s)$. In both cases the problem appears in the execution of Step 2, since we need to compute simple points \emph{that are radical} over $\F(s)$, i.e. points whose coordinates are over a root tower of $\F(s)$. In the following we see how to do that for the cases of genus 2 and 3, and some special cases of genus 1 and 4. Note that, in this section, we are slightly changing the notation since the parameters $\ot$ are now $(s,t)$.

\subsection{Genus 2}

\begin{theorem}[Case of genus 2]\label{theorem genus 2}
If $\ss$ has a genus 2 pencil of curves, then $\ss$ is radical.
\end{theorem}
\begin{proof}
Note that no simple point is needed in Algorithm \ref{alg curves 2<=g<=4}.
\end{proof}

\begin{example}\label{ex-3}
We consider the surface $\ss$ (see Fig. \ref{fig-2}) defined by
\[ \begin{array}{lcl}
 F(x,y,z) &=&y^2+2y^2z+2y^2z^6+y^3+xy-xyz-2yz^6x-4xy^2-2xy^2z-\\
  & & xy^2z^6+z^6x^3+zx^3y+x^3y^2\,.
\end{array} \]
$\ss$ has degree 9, but $\deg_{x}(f)=3$ and hence one can parametrize using case 1 in Section \ref{sec-radicals}.
\begin{figure}
  \centering
    \includegraphics[width=5cm]{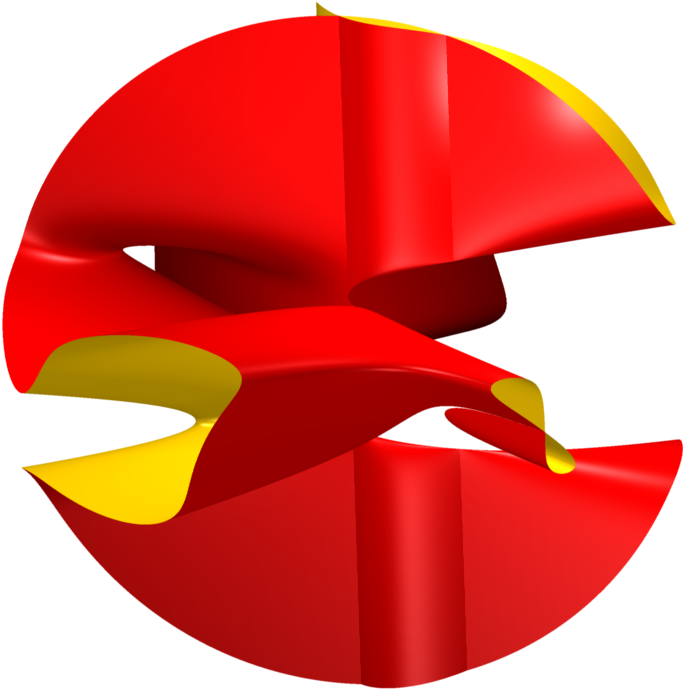}
  \caption{Surface in Example \ref{ex-3}}
  \label{fig-2}
\end{figure}
Nevertheless, we observe that $\ss$ contains the pencil of curves $\ccs$ (see some of the curves in Fig. \ref{fig-3}) defined by
\[
 g(x,y)=F(x,y,s)\in \F(s)[x,y].
\]
Note that $\ccs$, as a curve over the algebraic closure of $\F(s)$, has degree 5. Moreover, its singular (projective) locus is
\[
 \{(0: 1 : 0), (1: 0: 0), (0: 0: 1), (1: 1: 1) \}
\]
where all points are double. So the genus of $\ccs$ is 2, and hence Algorithm \ref{alg curves 2<=g<=4} is applicable. We consider the linear system of adjoints of degree $d-3=2$. Its (affine) defining polynomial is
\[
 h(x,y,t)=(-1-t)y+x+txy
\]
\begin{figure}
  \centering
    \includegraphics[width=5cm]{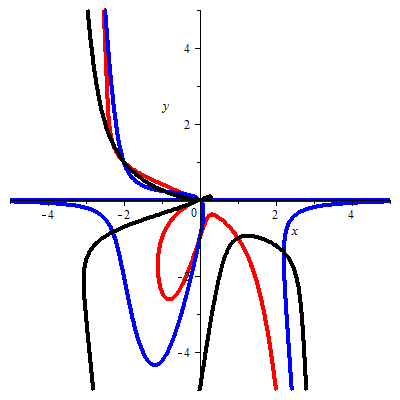}
  \caption{Some curves in the pencil of  Example \ref{ex-3}}
  \label{fig-3}
\end{figure}
Now, we compute the primitive part with respect to $t$ of the resultant of $g$ and $h$ w.r.t. $y$ and $x$ respectively, and get
\[ \begin{array}{ccl} S_1(t,s,x) & = & -t x^2-x^2 s^6 t^3+x^2 s t^2-2 s t x+3 x s^6t^2-t x+x s^6t^3+x \\
                                 &   & +2+s+t^2- s t^2-2 s^6 t-2 s^6 t^2+3 t \\
   S_2(t,s,y) & = & -t^3y^2+2yst^2-ys^6 t^3+ys^6t^2-3yt-2yt^2-y-2+st \\
     & & -s-t+2s^6 t
\end{array} \]
By theory $\deg_{x}(R_1),\deg_{y}(R_2)\leq 4$; indeed the degree is 2 in both cases. Now, computing by radicals the roots of $R_1$ and $R_2$ as polynomials in $x$ and $y$ respectively, we get the radical parametrization of $\ss$
\[
 \left( \frac {-2\,st+3\,s^6t^2-t+s^6t^3+1+\sqrt {A(t,s)}}{2t(1+s^6t^2-st)} \,,\right.
 \]
\[
 \left. - \frac {-2\,st^2+s^6t^3-s^6t^2+3\,t+2\,t^2+1+\sqrt {A(t,s)}}{2t^3}, s \right)
\]
where
\[ \begin{array}{ccl}
A(t,s) & = & 1+6\,t+4\,t^3+13\,t^2-4\,s^6t^3-4\,st^2-2\,s^6t^2+4\,s^7t^4+s^{12}t^4 \\
  & & -2\,s^{12}t^5+10\,s^6t^4+s^{12}t^6-16\,st^3+4\,s^6t^5-4\,s^7t^5-4\,st^4+4\,s^2t^4
\end{array} \]
\end{example}

\subsection{Genus 3}

Now we deal with the case of genus 3. Let $\overline{\ccs}$ denote the projective closure of $\ccs$ and let $\A_{d-3}(\overline{\ccs})$ be the system of adjoints to $\ccs$ of degree $d-3$. If we apply Algorithm \ref{alg curves g<=3}, $d$ simple points are required and if we apply Algorithm \ref{alg curves 2<=g<=4} we need one simple point. Now, we observe that $\dim(\A_{d-3}(\overline{\ccs}))=\gg-1$ and the number of simple intersection points in $\A_{d-3}(\overline{\ccs})\cap \overline{\ccs}$ is, in general, $2(\gg-1)=4$. Therefore this intersection contains, in general, 4 radical simple points on the curve. Taking one of them we generate a radical parametrization of $\ccs$, and hence of $\ss$. Thus, we have the following result.

\begin{theorem}[Case of genus 3]\label{theorem genus 3}
If $\ss$ has a genus 3 pencil of curves, then $\ss$ is radical.
\end{theorem}

Computationally, the question remains on how to compute one of these 4 radical simple points. The idea is as follows. Take an element in $\A_{d-3}(\overline{\ccs})$; for almost all selections it will work. Say that $M(x,y,w)$ is its defining polynomial. Then take $A_1(x,s)= \res_y(F(x,y,s),M(x,y,1))$ and $A_2(y,s)=\res_x(F(x,y,s),M(x,y,1))$. The roots of $A_1$ are the $x$-coordinates of the affine points in $\A_{d-3}(\overline{\ccs})\cap \overline{\ccs}$; similarly for $A_2$. So, crossing out the factors coming from the singularities of $\ccs$, we get two univariate polynomials (one in $x$ and the other in $y$) of degree at most 4. Solving them by radicals and recombining the results one gets the radical points.

\begin{example}\label{ex-4}
We consider the surface $\ss$ (see Fig. \ref{fig-3b}) defined by
\[ \begin{array}{lcl}
 F(x,y,z) & = & 2\,y-2\,y^2-y^2z^5-2\,y^3+2\,y^3z^5-y^4z^5+2\,y^4+z^5x^2y^2 \\
   & & +x^3+x^3y^2-2\,x^4-2\,x^4y+x^5
\end{array} \]
\begin{figure}
  \centering
    \includegraphics[width=5cm]{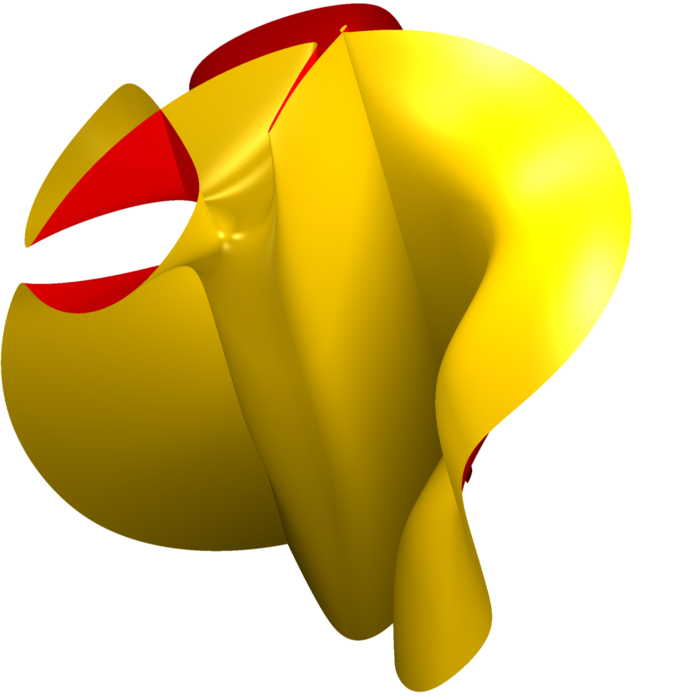}
  \caption{Surface in Example \ref{ex-4}}
  \label{fig-3b}
\end{figure}
$\ss$ has degree 9 and contains the pencil of curves $\ccs$ defined by
\[
 g(x,y)=F(x,y,s)\in \F(s)[x,y].
\]
Note that $\ccs$, as a curve over the algebraic closure of $\F(s)$, has degree 5. Moreover, its singular (projective) locus is
\[
 \{ (0:1:1),(1:1:0),(1:0:1) \}
\]
where all points are double. So the genus of $\ccs$ is 3. We apply the ideas in Theorem \ref{theorem genus 3}. The adjoints of degree $d-3=2$ are defined by
\[ \begin{array}{rcl}
 H(x,y,w,\lambda_1,\lambda_2,\lambda_3) & = & \left( -\lambda_1-\lambda_2 \right) w^2 + \left(\lambda_3+2\lambda_1+\lambda_2 \right) yw + \left( -\lambda_3-\lambda_1 \right) y^2 \\
  & & + \lambda_2 xw+\lambda_3 xy+\lambda_1 x^2.
\end{array} \]
We take a curve in the system. For instance
\[
 M(x,y,1):=H(x,y,1,-1,1,0)=-y+y^2+x-x^2.
\]
Now, computing the intersection of the curve defined by $M(x,y,1)$ and $\ccs$ we get, at most, 4 (affine) radical simple points on $\ccs$. For this purpose, we compute
\[ \begin{array}{lcl}
 A_1(x,t) & = & \res_y(F(x,y,s),M(x,y,1)), \\
 A_2(y,t) & = & \res_x(F(x,y,s),M(x,y,1)).
\end{array} \]
We get
\[ \begin{array}{lcl}
 A_1 & = & 4\,x^3 \left( x+1 \right) \left( x-1 \right)^2 \left( 2\,s^5x^2-x^2-2\,x-s^5x+2 \right) \\
 A_2 & = & 4\,y^3 \left( y-2 \right) \left( -1+y \right)^2 \left( 2\,y^2s^5-y^2-2\,y-ys^5+2 \right)
\end{array} \]
which generate the following affine radical points in $\ccs$:
\[ \left\{ (-1,2),(0,0),(0,1),(1,0),\right. \]
\[
 \left(\frac{1}{2}\,{\frac {2+{s}^{5}-\sqrt {12-12\,{s}^{5}+{s}^{10}}}{2\,{s}^{5}-1}},\frac{1}{2}\,{\frac {2+{s}^{5}-\sqrt {12-12\,{s}^{5}+{s}^{10}}}{2\,{s}^{5}-1}}\right),
\]
\[
 \left. \left(\frac{1}{2}\,{\frac {2+s^5+\sqrt {12-12\,s^5+s^{10}}}{2\,s^5-1}},\frac{1}{2}\,{\frac {2+s^5+\sqrt {12-12\,s^5+s^{10}}}{2\,s^5-1}} \right) \right\}.
\]
Note that $(1,0)$ and $(0,1)$ are singular. We now use one of these simple points, say $(-1,2)$, to reduce the dimension of the system of adjoints down to one. We get that the defining polynomial of the 1-dimensional resulting subsystem is
\[
 H^*(x,y,w,t)= -w^2-tw^2+ \left( 2w+tw \right) y-y^2+txw+x^2
\]
In this situation, the theory ensures that the intersection of $H^*$ and $\overline{\ccs}$ leaves, at most, 4 intersections points different to the singularities and to the simple point $(-1,2)$. Computing these intersection points we reach the radical parametrization. For this purpose, we determine the primitive part with respect to $t$ of the resultant of $H^*(x,y,1,t)$ and $F(x,y,s)$ w.r.t. $y$ and $x$ respectively to get
\[ \begin{array}{lcl}
 S_1(x,t,s) & = & 8\,x^3+10\,x^3t-2\,s^5tx^3+x^3t^2+18\,x^2t+4\,x^2-5\,s^5t^2x^2 \\
   & & +12\,t^2x^2-4\,x^2s^5t-4\,xs^5t^3+8\,tx+8\,t^3x+18\,t^2x-2\,xs^5t \\
   & & -6\,xs^5t^2-2\,s^5t^3+6\,t^3+2\,t^4+4\,t^2-s^5t^4-s^5t^2 \\
 S_2(y,t,s) & = & -8\,y^3-10\,y^3t+2\,y^3s^5t-y^3t^2+3\,t^3y^2+21\,y^2t^2-2\,y^2s^5t \\
   & & +20\,y^2+36\,y^2t-y^2s^5t^2-16\,y-20\,t^3y-45\,yt^2-3\,t^4y \\
   & & -42\,yt+t^5+16\,t+4+19\,t^3+7\,t^4+25\,t^2
\end{array}\]


As expected, $\deg_x(R_1)=\deg_y(R_2)=3$. Solving by radicals these polynomials we get the parametrization
\[ \left(\frac{1}{6}\,{\frac {\Sigma}{ \left( -10\,t+2\,{s}^{5}t-8-{t}^{2} \right)
\sqrt [3]{\Gamma }}},\frac{1}{6}\,{\frac {\Sigma}{ \left( -10\,t+2\,{s}^{5}t-8
-{t}^{2} \right) \sqrt [3]{\Gamma }}},s\right) \]
where
\[ \begin{array}{lcl}
\Delta &= &  - ( 2+t )^{2} ( -108\,{t}^{8}-304\,{t}^{2}-1284\,{t}^
{5}{s}^{10}+2932\,{t}^{6}{s}^{5}+252\,{s}^{15}{t}^{4}\\ & &+144\,{s}^{15}{t}
^{2}+264\,{s}^{15}{t}^{3}+132\,{s}^{15}{t}^{5}-662\,{s}^{10}{t}^{6}+
876\,{t}^{7}{s}^{5}+36\,{s}^{15}{t}^{6}\\ &&-204\,{s}^{10}{t}^{7}+108\,{t}^
{8}{s}^{5}+32\,{s}^{15}t-64\,{s}^{5}-3168\,{t}^{3}+256\,{s}^{5}t+128\,
t\\ & & +5240\,{t}^{4}{s}^{5}+3384\,{t}^{3}{s}^{5}+1408\,{s}^{5}{t}^{2}+16\,{
s}^{10}-27\,{t}^{8}{s}^{10}+4\,{t}^{7}{s}^{15}\\ & & -1616\,{s}^{10}{t}^{3}-
192\,{s}^{10}t-888\,{s}^{10}{t}^{2}-1731\,{s}^{10}{t}^{4}+5156\,{t}^{5
}{s}^{5}-6944\,{t}^{5}\\ & & -3612\,{t}^{6}-968\,{t}^{7}-6880\,{t}^{4})\\
\Gamma & =& 512+216\,{t}^{8}-10368\,{t}^{2}-5256\,{t}^{5}{s}^{10}+2700\,{t}^{6}{s}
^{5}+96\,{s}^{15}{t}^{4}+64\,{s}^{15}{t}^{3}\\ & & +48\,{s}^{15}{t}^{5}-1944
\,{s}^{10}{t}^{6}-72\,{t}^{7}{s}^{5}+8\,{s}^{15}{t}^{6}-288\,{s}^{10}{
t}^{7}-108\,{t}^{8}{s}^{5}\\  & &-10368\,{t}^{3}+768\,{s}^{5}t-2304\,t+22176
\,{t}^{4}{s}^{5}+21312\,{t}^{3}{s}^{5}+9024\,{s}^{5}{t}^{2}\\ & & +24\,{t}^{2
}\sqrt {3}\sqrt {\Delta}{s}^{5}-12\,{t}^{3}\sqrt {3}\sqrt {\Delta}-96
\,t\sqrt {3}\sqrt {\Delta}-120\,{t}^{2}\sqrt {3}\sqrt {\Delta}\\ & &-5088\,{
s}^{10}{t}^{3}-1344\,{s}^{10}{t}^{2}-7248\,{s}^{10}{t}^{4}+11664\,{t}^
{5}{s}^{5}+7056\,{t}^{5}\\ & & +4752\,{t}^{6}+1512\,{t}^{7}+576\,{t}^{4} \\
\Sigma & = &{\Gamma }^{2/3}-192\,t-1008\,{t}^{2}-1296\,{t}^{3}-600\,{t}^{4}-96\,{t
}^{5}+272\,{s}^{5}{t}^{2}+456\,{t}^{3}{s}^{5}\\ & &+64\,{s}^{5}t+264\,{t}^{4
}{s}^{5}+16\,{s}^{10}{t}^{3}+16\,{s}^{10}{t}^{2}+4\,{s}^{10}{t}^{4}+48
\,{t}^{5}{s}^{5}+64\\ && +36\,\sqrt [3]{\Gamma }t+8\,\sqrt [3]{\Gamma }-10\,
\sqrt [3]{\Gamma }{s}^{5}{t}^{2}+24\,\sqrt [3]{\Gamma }{t}^{2}-8\,
\sqrt [3]{\Gamma }{s}^{5}t
\end{array}
\]
 \end{example}

\subsection{Genus 1}

For the case of genus 1, only Algorithm \ref{alg curves g<=3} is applicable and, in Step 2, $d-2$ simple points are required. So let us assume that $\ccs$ has a radical double point $P$. Then, intersecting $\ccs$ with a line defined over $\F$ and passing through the double point, one can take a family of $d-2$ $\F(s)$-conjugate points (see Def. 3.15 in \cite{libro}) that can be used in Step 2. Moreover, the subsystem of adjoints $\A_{d-2}^*$ has defining polynomial over $\F(s)$ (see Lemma 3.19 in \cite{libro}). Therefore, we get the following theorem:

\begin{theorem}[Case of genus 1; first part]\label{theorem genus 1}
If $\ss$ has a genus 1 pencil of curves with a double radical point, then $\ss$ is  radical.
\end{theorem}

\begin{example}\label{ex-5}
We consider the surface $\ss$ (see Fig. \ref{fig-4}) defined by
\[ \begin{array}{lcl}
 F(x,y,z) & = & {y}^{2}+32\,{y}^{2}z+80\,{y}^{2}{z}^{2}+{y}^{3}-\frac{5}{2}\,xy-55\,xyz-132\,xy{z}^{2}-\frac{9}{2}\,x{y}^{2} \\
   & & -22\,x{y}^{2}z-36\,x{y}^{2}{z}^{2}+zx{y}^{3}+\frac{3}{2}\,{x}^{2}+24\,{x}^{2}z+54\,{x}^{2}{z}^{2}+\frac{13}{2}\,{x}^{2}y \\
   & & +32\,{x}^{2}yz+42\,{x}^{2}y{z}^{2}-3\,{x}^{3}-12\,{x}^{3}z-9\,{x}^{3}{z}^{2}+{z}^{2}{x}^{3}{y}^{2}.
\end{array} \]
$\ss$ has degree 7, but $\deg_{z}(F)=2$ and hence one can parametrize using case 1 in Section \ref{sec-radicals}.
\begin{figure}
  \centering
    \includegraphics[width=5cm]{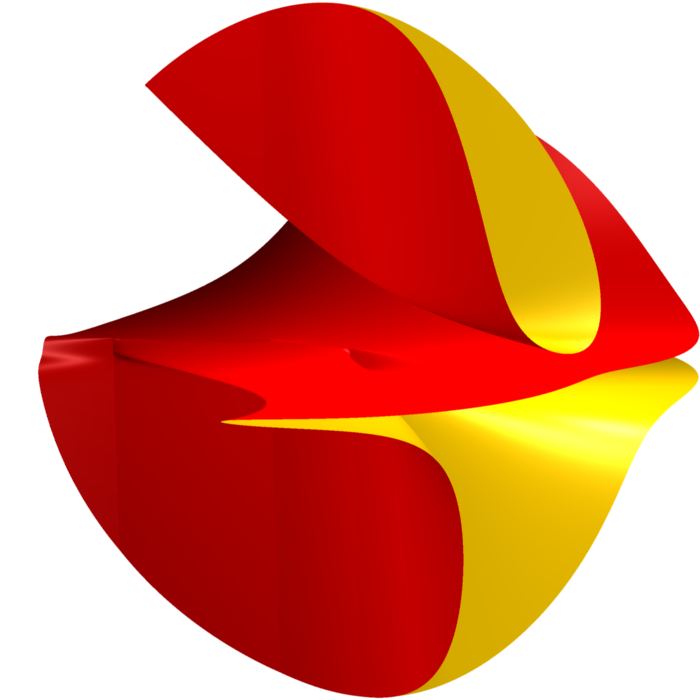}
  \caption{Surface in Example \ref{ex-5}}
  \label{fig-4}
\end{figure}
Nevertheless, we observe that $\ss$ contains the pencil of curves $\ccs$ defined by
\[
 g(x,y)=F(x,y,s)\in \F(s)[x,y].
\]
Note that $\ccs$, as a curve over the algebraic closure of $\F(s)$, has degree 5. Its singular (projective) locus is
\[
 \{(0: 1: 0), (1: 0: 0), (0: 0: 1), (1: 1: 1), (2: 3: 1)\}
 \]
where all points are double. So the genus of $\ccs$ is 1. Moreover it has rational double points. Thus, Theorem \ref{theorem genus 1} is applicable. So we apply Algorithm \ref{alg curves g<=3} taking, in Step 2, $d-2=5$ points on a family of conjugate points. First we consider the linear system of adjoints of degree $d-2=5$. It is defined by
\[ \begin{array}{lcl}
 H(x,y,w) & = & \left( -7\,a_{{0,2}}-4\,a_{{1,1}}-16\,a_{{1,2}}-2\,a_{{2,0}}-10\,a_{{2,1}} \right) y{w}^{2}+a_{{0,2}}{y}^{2}w+ \\
   & & \left( 6\,a_{{0,2}}+3\,a_{{1,1}}+15\,a_{{1,2}}+a_{{2,0}}+9\,a_{{2,1}} \right) x{w}^{2}+a_{{1,1}}xyw \\
   & & +a_{{1,2}}x{y}^{2}+a_{{2,0}}{x}^{2}w+a_{{2,1}}{x}^{2}y.
\end{array} \]
Now, cutting with the line $x+y=0$, we get the family of 5 conjugate points
\[
 \mathcal{F}=\{( \alpha:-\alpha: 1)\,|\,m(\alpha)=0\}
\]
where $m\in \F(s)[X]$ is
\[
 m(X)=5+111s+266s^2-15 X-66Xs-87Xs^2-X^2s+X^3s^2.
\]
Then, requiring the remainder of $H(X,-X,1)$ with respect to $m(X)$ to be zero, we get the conditions
\[
\begin{array}{r@{\ }lll}
  \{ & a_{0, 2} = a_{0, 2}, & a_{1, 1} = -5 a_{1, 2}-4 a_{0, 2}, & a_{1, 2} = a_{1, 2}, \\
     & a_{2, 0} = 10a_{1, 2}+4a_{0, 2}, & a_{2, 1}= -2a_{1, 2} \ \}
\end{array}\]
which provide the following 1-dimensional linear subsystem of adjoints:
\[ \begin{array}{lcl}
 H^*(x,y,w) &= &  \left( 1+4\,t \right) y{w}^{2}+{y}^{2}w+ \left( -2-8\,t \right) x{w}^{2}+ \left( -5\,t-4 \right) xyw \\
   & &+tx{y}^{2}+ \left( 10\,t+4 \right) {x}^{2}w-2\,t{x}^{2}y
\end{array} \]
Now, we compute the primitive part with respect to $t$ of the resultant of $g$ and $H(x,y,1)$ w.r.t. $y$ and $x$ respectively to get
\[ \begin{array}{ccl}
 S_1(t,s,x) & = & 16\,{s}^{2}{t}^{3}{x}^{2}+320\,{s}^{2}{t}^{3}x+20\,{s}^{2}{t}^{2}{x}^2-12\,s{x}^{2}{t}^{3}+176\,{s}^{2}{t}^{2}x \\
    & & +4\,{x}^{2}{s}^{2}t+112\,s{t}^{3}x-3\,{t}^{3}{x}^{2}+320\,{t}^{2}{s}^{2}+40\,{s}^{2}tx+28\,s{t}^{2}x \\
    & & +16\,{t}^{3}x+160\,t{s}^{2}+4\,x{s}^{2}+128\,{t}^{2}s-16\,{t}^{3}+4\,{t}^{2}x+20\,{s}^{2} \\
    & & +64\,ts-8\,{t}^{2}+8\,s-t \\
 S_2(t,s,y) & = & -160\,{s}^{2}{t}^{3}y+832\,{t}^{3}{s}^{2}-72\,{t}^{2}{s}^{2}y-56\,s{t}^3y+2\,{y}^{2}{t}^{2}s-2\,{y}^{2}{t}^{3} \\
   & &+768\,{t}^{2}{s}^{2}+256\,{t}^{3}s-46\,{t}^{2}sy+12\,{t}^{3}y+228\,t{s}^{2}+2\,y{s}^{2} \\
   & & +256\,{t}^{2}s-8\,yts-16\,{t}^{3}+3\,y{t}^{2}+22\,{s}^{2}+80\,ts-8\,{t}^{2}+8\,s-t
\end{array} \]
By theory $\deg_{x}(R_1),\deg_{y}(R_2)\leq 4$; indeed the degree is 2 in both cases. Now, computing by radicals the roots of $R_1$ and $R_2$ as polynomials in $x$ and $y$ respectively, we get the radical parametrization of $\ss$:
\[
 \left( {\frac { \left( -2\,{t}^{2}-40\,{t}^{2}{s}^{2}-14\,{t}^{2}s-12\,t{s}^{2}-2\,{s}^{2}+\sqrt {\Delta} \right)  \left( 1+4\,t \right) }{t \left( -12\,{t}^{2}s+4\,{s}^{2}+20\,t{s}^{2}-3\,{t}^{2}+16\,{t}^{2}{s}^{2} \right) }},\right.
\]
\[
 \left. {\frac { \left( -3\,{t}^{2}+40\,{t}^{2}{s}^{2}+14\,{t}^{2}s+8\,ts+8\,t{s}^{2}-2\,{s}^{2}+\sqrt {\Delta} \right) \left( 1+4\,t \right) }{4{t}^{2} \left( s-t \right) }},s\right)
\]
where
\[ \begin{array}{ccl}
 \Delta(t,s) & = & 4\,{s}^{4}+24\,{t}^{3}s+224\,{t}^{3}{s}^{2}+12\,{t}^{2}{s}^{2}+44\,s{t}^{4}+{t}^{4}+372\,{t}^{4}{s}^{2} \\
   & & +1600\,{t}^{4}{s}^{4}+1120\,{t}^{4}{s}^{3}+448\,{t}^{3}{s}^{3}+640\,{t}^{3}{s}^{4}-96\,{t}^{2}{s}^{4} \\
   & & -104\,{t}^{2}{s}^{3}-32\,t{s}^{3}-32\,t{s}^{4}.
\end{array} \]
\end{example}

Theorem \ref{theorem genus 1} can be generalized as follows. If the pencil of genus 1 curves has an $r$-fold radical point and $r-2$ simple radical points (note that when $r=2$ we are in the situation of Theorem \ref{theorem genus 1}), we can proceed as follows. In Step 2 of Algorithm \ref{alg curves g<=3} we need $d-2$ simple points. Then, intersecting $\ccs$ with a line defined over $\F$ and passing through the $r$-fold point, one can take a family of $d-r$ $\F(s)$-conjugate points. Now, using the $r-2$ simple rational points we get a 1-dimensional subsystem of adjoints $\A_{d-2}^*$, definable over $\F(s)$, that parametrizes the pencil by radicals. Therefore, we get the following theorem.

\begin{theorem}[Case of genus 1; second part]\label{theorem genus 1b}
If $\ss$ has a genus 1 pencil of curves with an $r$-fold radical point and $r-2$ radical simple points, then $\ss$ is radical.
\end{theorem}

There is still another situation where we can solve by radicals a surface that has a pencil of genus 1 curves, albeit in a different fashion (see for example \cite[Prop. IV.4.6]{Hartshorne}). Let $P$ be a radical regular point of $\ccs$, i.e. whose coordinates lie in a root tower over $\F(s)$. Then, by the Riemann-Roch theorem \cite[Th. IV.1.3]{Hartshorne} we have
\[
 \dim(\cL(nP))=n \,, \quad n\geq1
\]
where $\cL(nP)$ is the vector space of rational functions defined on the curve such that the only possible pole is $P$ with order at most $n$. The computation of bases for several of these spaces (see \cite{Hess}) will provide a low degree equation:
\begin{itemize}
 \item $\cL(2P)=\langle 1,f\rangle$ where $f$ has a double pole at $P$,
 \item $\cL(3P)=\langle 1,f,g\rangle$ where $g$ has a triple pole at $P$.
\end{itemize}
Now, we have $\cL(6P)=\langle 1,f,g,f^2,fg,g^2,f^3\rangle$ where the last two elements have pole order exactly 6 at $P$. Since we have 7 elements in a vector space of dimension 6, there exists a nontrivial linear relation between these functions. It follows that the map $Q\mapsto(f(Q),g(Q))$ is birrationally equivalent to its image, whose equation is precisely the linear relation just mentioned. In this way we reduce $\ccs$ to a cubic curve, and hence to a curve parametrizable by radicals; see \cite{SeSevilla}.

\begin{theorem}[Case of genus 1; third part]\label{theorem genus 1c}
If $\ss$ has a genus 1 pencil of curves with a regular radical point, then $\ss$ is radical.
\end{theorem}

\subsection{Genus 4}

In Step 2 of Algorithm \ref{alg curves 2<=g<=4}, two simple points are required. Therefore, the following theorem holds.

\begin{theorem}[Case of genus 4; first part]\label{theorem genus 4}
If $\ss$ has a genus 4 pencil of curves with two simple radical points, then $\ss$ is radical.
\end{theorem}

\begin{example}\label{ex-6}
We consider the surface $\ss$ (see Fig. \ref{fig-5}) defined by
\[
 F(x,y,z) = -1+y^3-3x-xz^7-xz+xy^2+xy^3+x^3y^2.
\]
$\ss$ has degree 8, but $\deg_{x}(f)=3$ and hence one can parametrize using case 1 in Section \ref{sec-radicals}.
\begin{figure}
  \centering
    \includegraphics[width=5cm]{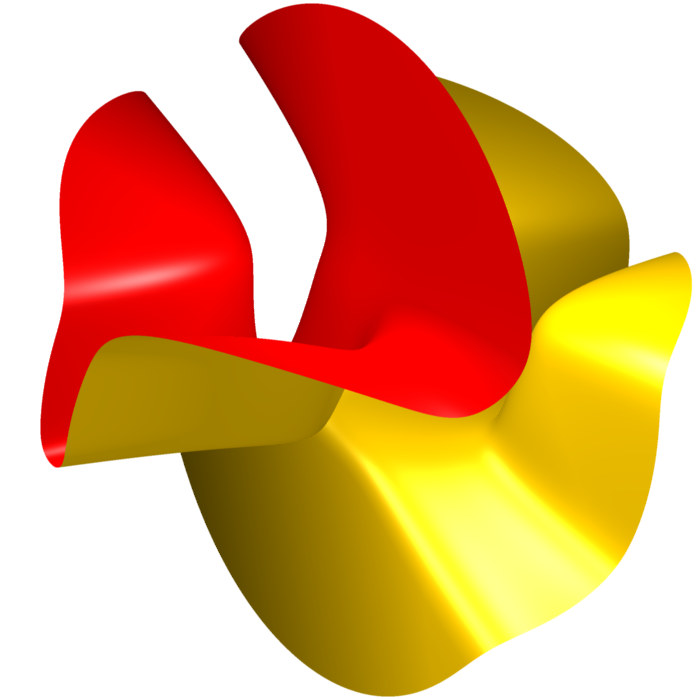}
  \caption{Surface in Example \ref{ex-6}}
  \label{fig-5}
\end{figure}
Nevertheless, we observe that $\ss$ contains the pencil of curves $\ccs$ defined by
\[
 g(x,y)=F(x,y,s)\in \F(s)[x,y].
\]
As a curve over the algebraic closure of $\F(s)$, $\ccs$ has degree 5. Its singular (projective) locus is
\[
 \{(0: 1: 0), (1: 0: 0) \}
\]
where all points are double. So the genus of $\ccs$ is 4. Moreover $\ccs$ has the following two radical simple points:
\[
 \{ (\sqrt {{s}^{7}+s+1}:1:1),(-\sqrt {{s}^{7}+s+1}:1:1) \}.
\]
We apply Algorithm \ref{alg curves 2<=g<=4} to $\ccs$. In Step 1 we get the $d-3=2$ adjoints. They are defined by
\[
 a_{{0,0}}{w}^{2}+a_{{0,1}}yw+a_{{1,0}}xw+a_{{1,1}}xy.
\]
Now, we consider the linear subsystem obtained by forcing the adjoints to pass through the radical simple points. We get
\[
 H^*(x,y,w)=-{w}^{2}+yw-txw+txy.
\]
In this situation, the theory ensures that the intersection of $H^*$ and $\overline{\ccs}$ leaves, at most, 4 intersection points different to the singularities and to the 2 chosen simple radical points. Computing these intersection points we reach the radical parametrization. For this purpose, we determine the primitive part with respect to $t$ of the resultant of $H^*(x,y,1,t)$ and $F(x,y,s)$ w.r.t. $y$ and $x$ respectively, to get
\[ \begin{array}{lcl}
 S_1(x,t) & = & - \left( 1+tx \right)^{3} \\
 S_2(y,t) & = & -{t}^{3}+{t}^{3}{y}^{3}+3\,{t}^{2}+{t}^{2}{s}^{7}+{t}^{2}s-{t}^{2}{y}^{2}-{t}^{2}{y}^{3}-{y}^{2}
\end{array} \]
Solving by radicals we get the radical parametrization
\[
 \left( \frac{-1}{t},\frac{1}{6}\,{\frac {{\Gamma}^{2/3}+4\,{t}^{4}+8\,{t}^{2}+4+2\,\sqrt [3]{\Gamma}{t}^{2}+2\,\sqrt [3]{\Gamma}}{{t}^{2} \left( t-1 \right) \sqrt [3]{\Gamma}}},s \right)
\]
where
\[ \begin{array}{lcl}
 \Delta & = & -12+4\,t-4\,s+12\,{t}^{3}-4\,{s}^{7}-36\,{t}^{2}-36\,{t}^{4}+231\,{t}^{6}+270\,{t}^{8}{s}^{7}\\ & &-378\,{t}^{7}{s}^{7}+158\,{t}^{6}{s}^{7}+270\,{t}^{8}s-378\,{t}^{7}s+158\,{t}^{6}s-54\,{t}^{9}{s}^{7}-54\,{t}^{9}s \\ & &+27\,{t}^{8}{s}^{14}+54\,{t}^{8}{s}^{8}-12\,{t}^{4}{s}^{7}-54\,{t}^{7}{s}^{14}-108\,{t}^{7}{s}^{8}+27\,{t}^{6}{s}^{14}+54\,{t}^{6}{s}^{8} \\
   & &+27\,{t}^{8}{s}^{2}-12\,{t}^{4}s-54\,{t}^{7}{s}^{2} +27\,{t}^{6}{s}^{2}-216\,{t}^{9}+594\,{t}^{8}-644\,{t}^{7}+27\,{t}^{10} \\
   & &+12\,{t}^{5}-12\,{t}^{2}{s}^{7}-12\,{t}^{2}s \\
 \Gamma & = & 108\,{t}^{9}-540\,{t}^{8}+756\,{t}^{7}-316\,{t}^{6}-108\,{t}^{8}{s}^{7}+216\,{t}^{7}{s}^{7}-108\,{t}^{6}{s}^{7} \\
   & &-108\,{t}^{8}s+216\,{t}^{7}s-108\,{t}^{6}s+24\,{t}^{4}+24\,{t}^{2}+8+12\,\sqrt {3}\sqrt {\Delta}{t}^{4} \\
   & & -12\,\sqrt {3}\sqrt {\Delta}{t}^{3}
\end{array} \]
\end{example}

In addition to the results above, with a different approach (also based on adjoints) we can solve the non-hyperelliptic genus 4 case without extra conditions. It is known by Brill-Noether theory (see \cite[Chapter V]{ACGH} for example) that genus 4 curves are trigonal (they admit a $3:1$ map to $\pj^1$). A general characterization of trigonality was given by Enriques, Petri and Babbage (see \cite{Saintdonat} for a modern account), here we develop the relatively simple case of genus 4. See the end of the section for comments on the hyperelliptic case.


The image of a non-hyperelliptic genus 4 curve by its canonical embedding (which is given by a basis of the adjoint space) is a degree 6 curve in $\pj^3$. By the aforementioned result, the intersection of all the quadric hypersurfaces containing the canonical image is a scroll, and any of its rulings determines a $3:1$ map; but the ambient dimension is 3, so the scroll itself is the only quadric hypersurface. Since the canonical embedding has as its coordinates a basis of the space of adjoints \cite[Chapter VI, \S 6]{walker},
\[
 \phi_C\colon C\to\pj^3 \,,\quad \phi_C(p)=(a_0(p):a_1(p):a_2(p):a_3(p)),
\]
the scroll can be calculated explicitly by finding a homogeneous degree 2 relation
\[
 F(x_0:x_1:x_2:x_3) \,, \deg(F)=2 \quad\mbox{such that}\quad F(a_0:a_1:a_2:a_3)=0.
\]
It remains to calculate a ruling. We have two cases, which we discern by computing the singularities of the scroll, i.e. solving a linear system of four equations in four variables. Either it has one singularity, and the coordinates of the singularity are rational in the coefficients of the linear system (thus of the adjoints), or it is regular.

\begin{itemize}
 \item If the scroll is singular point, it is a cone. The lines that provide its ruling can be calculated by projecting the cone from its vertex (we obtain a plane conic), and then joining the vertex to each point of the conic. Finding a point and the parametrization of the conic involves, in the worst case, introducing one square root.
 \item If the scroll is not singular, we calculate one point in it (this can be done, at worst, with one extra square root) and its tangent space. The intersection of both is a reducible conic, thus a pair of lines. Let $L$ be one of them. The projection from $L$ provides a ruling: each plane containing $L$ cuts the scroll in a union of $L$ and another line.
\end{itemize}

Therefore, we have the following result.

\begin{theorem}[Case of genus 4; second part]\label{theorem genus 4;second part}
If $\ss$ has a pencil of genus 4 non-hyperelliptic curves, then $\ss$ is radical.
\end{theorem}

\begin{example}\label{ex-trigonal}
We consider the surface $\ss$ (see Fig. \ref{fig-trigonal}) defined by
\[ \begin{array}{lcl}
 F(x,y,z) & = & 2\,{z}^{5}-y{z}^{4}+2\,{y}^{2}{z}^{3}+14\,{y}^{3}{z}^{2}+8\,{y}^{4}z-{y}^{5}+xy{z}^{3}+x{y}^{2}{z}^{2} \\
   & & +2\,x{y}^{3}z+2\,x{y}^{4}-2\,{x}^{2}y{z}^{2}-9\,{x}^{2}{y}^{2}z+2\,{x}^{2}{y}^{3}-{x}^{3}{z}^{2} \\
   & & -{x}^{3}{y}^{2}-{x}^{4}z-7\,{x}^{4}y+5\,{x}^{5}.
\end{array} \]
\begin{figure}
  \centering
    \includegraphics[width=5cm]{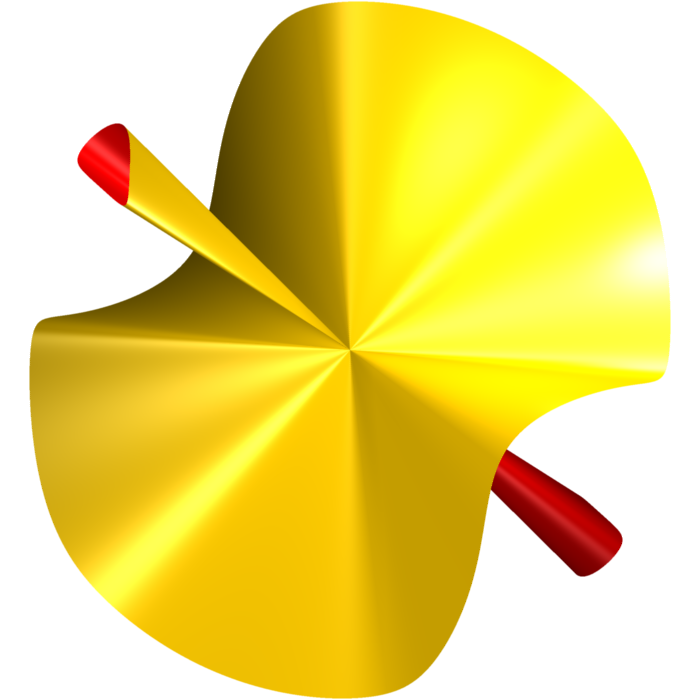}
  \caption{Surface in Example \ref{ex-trigonal}}
  \label{fig-trigonal}
\end{figure}
We observe that $\ss$ contains the pencil of curves $\ccs$ defined by
\[
 g(x,y)=F(x,y,s)\in \F(s)[x,y].
\]
As a curve over the algebraic closure of $\F(s)$, $\ccs$ has degree 5. Its singular (projective) locus is
\[
 \{(0: -s: 1), (s: 0: 1) \}
\]
where all points are double. So the genus of $\ccs$ is 4, thus it is trigonal and we can apply the procedure above (but note that the pencil of curves through one of the singular points would already produce a $3:1$ map). Its adjoints are quadrics passing through the two points above, and a basis is
\[
 \langle y^2+ys, y^2+x^2-s^2, xy, x^2-xs \rangle.
\]
A nontrivial quadratic relation between these four generators is given by the polynomial in $u_0,u_1,u_2,u_3$
\[
 u_0u_2+u_0u_3-u_1u_2+u_2u_3-u_2^2.
\]
It is immediate to check that the only common zero of its partial derivatives is the origin, therefore it is a nonsingular quadric surface in $\pj^3$. To produce a ruling, we choose a point in it, say $(0:1:0:0)$. The tangent plane at this point is $u_2=0$. Setting this in the equation of the surface we obtain
\[
 u_0u_3=0
\]
producing the two lines $u_2=u_3=0$ and $u_2=u_0=0$. We choose the first line as our $L$. The planes containing $L$ are of the form $\alpha u_2+\beta u_3=0$ with $(\alpha:\beta)\in\pj^1$. Each one of them intersects the surface in
\[
 L \quad\cup\quad \{\alpha u_2+\beta u_3=0 \,,\, (\beta-\alpha)u_0-\beta u_1-(\alpha+\beta)u_2=0\}.
\]
Then, the ruling is given by
\[
 (u_0:u_1:u_2:u_3)\mapsto(\alpha:\beta)=(u_0-u_1-u_2:u_0+u_2).
\]
Composing this with the embedding by adjoints we obtain $\ccs\to\pj^1$ given by
\[
 (x,y)\mapsto(ys+s^2-x^2-xy:ys+y^2+xy)
\]
We know that this map is $3:1$. Indeed, if we define
\[
 H(x,y,s,t)=(ys+s^2-x^2-xy)-t(ys+y^2+xy)
\]
then the primitive parts with respect to $t$ of the resultants of $F(x,y,s)$ and $H(x,y,s,t)$ w.r.t. $y$ and $x$ are
\[ \begin{array}{lcl}
 S_1(x,s,t) & = & \left( -{t}^{2}+{t}^{3}-8\,{t}^{5}-1+3\,{t}^{4}-2\,t \right) {y}^{3} \\
  & & +\left( -11\,s{t}^{3}+30\,s{t}^{4}+10\,s-s{t}^{2} \right) {y}^{2} \\
  & & +\left( 13\,{s}^{2}{t}^{2}-40\,{s}^{2}{t}^{3}+15\,{s}^{2}-2\,{s}^{2}t
 \right) y+4\,{s}^{3}+20\,{s}^{3}{t}^{2}-6\,{s}^{3}t \\
 S_2(y,s,t) & = & \left( -{t}^{2}+{t}^{3}-8\,{t}^{5}-1+3\,{t}^{4}-2\,t \right) {x}^{3} \\
  & & +\left( 2\,s{t}^{4}-4\,st-2\,s{t}^{3}+3\,s{t}^{2}+3\,s-6\,s{t}^{5}
 \right) {x}^{2} \\
  & & +\left( 12\,{s}^{2}{t}^{2}-{s}^{2}{t}^{3}-4\,{s}^{2}{t}^{5}-3\,{s}^{2}+14\,{s}^{2}t \right) x \\
  & & +{s}^{3}{t}^{4}-8\,{s}^{3}t-2\,{s}^{3}{t}^{3}-14\,{s}^{3}{t}^{2}+{s}^{3}-2\,{t}^{5}{s}^{3}
\end{array} \]
from which we solve by radicals to calculate the parametrization.
\end{example}

We finish this section with the observation that for genus 4 in general the surface is radical, though we do not provide an algorithm for the hyperelliptic case. The idea is that every hyperelliptic curve can be transformed into one of the form $y^2=f(x)$ where $\deg(f)$ is $2\gg+2$ or $2\gg+1$; and there exist algorithms to detect whether a curve is hyperellipic and put it in that form, in such a way that at most a quadratic extension of $\F(s)$ is needed, thus we can parametrize such curves by radicals using Case 1 of Section \ref{sec-radicals}. Thus we have the following result.

\begin{theorem}[Case of genus 4; third part]\label{theorem genus 4;third part}
If $\ss$ has a pencil of genus 4 curves, then $\ss$ is radical.
\end{theorem}

\section{The Role of Radical Parametrizations in Geometric Constructions}\label{sec-offsets}

In many geometric constructions, as for instance offset or conchoidal curves/surfaces, one observes that even though one starts from a rational parametrization the generated object is not, in general, parametrizable by means of rational functions; see \cite{PP}, \cite{ASS}, \cite{ASS2} for offsets and \cite{PGS}, \cite{SeSe} for conchoids. However, if one starts from a radical parametrization the offset (similarly the conchoid) is radical. In other words, the class of radical curves/surfaces (i.e curves/surfaces having a radical parametrization) are invariant under offsetting and conchoidal constructions.

In this section, we state the above claim and we enlarge the family of potential geometric constructions that preserves the class of radical parametrizations. For this purpose, we see the geometric construction from the perspective of algebraic geometry by considering varieties of incidence. To clarify ideas, let us start with two examples (offsets and conchoids) and afterwards we present the more general frame.

\subsection{Offset construction}

Let $F(x_1,x_2,x_3)$ be the defining polynomial of an irreducible surface $\ss$ and $\delta\in \F$ a non-zero element. Then, we consider the incidence variety (where $\xx=(x_1,x_2,x_3)$ and $\yy=(y_1,y_2,y_3)$)
\[
 \bb=\left\{ (\yy,\xx,\lambda,W) \in \F^3 \times \F^3 \times \F \times \F \left/
   \begin{array}{l}
        F(\xx)=0   \\
        \sum_{i=1}^{3}(x_{i}-y_{i})^2=\delta^2 \\
        \yy=\xx+\lambda \nabla(F)(\xx) \\
        W \prod_{i=1}^{3} \frac{\partial F}{\partial x_i}(\xx) =1
   \end{array}
  \right. \right \}.
\]
The first equation ensures that $\xx \in \ss$, the second equation is the sphere center at $\xx$ with radius $\delta$, the third equation is the normal line to $\ss$ at the point $\xx$, and the forth equation ensures that the gradient vector is not the zero vector, i.e. that $\xx$ is not singular. We observe that all $\yy$, being part of a tuple in $\bb$, are indeed in the offset. Motivated by this fact, we consider the diagram
\[\xy
 (0,0)*+++{\bb \subset \F^3\times\F^3\times\F\times\F}="B";
 (-15,-20)*++{\pi_1(\bb)\subset\F^3}="k2";
 (15,-20)*++{\ss\subset\F^3}="k";
 {\ar_{\displaystyle \pi_1} "B"; "k2"};
 {\ar^{\displaystyle \pi_2} "B"; "k"};
 (50,-10)*{\mbox{(Offset Incidence Diagram)}};
\endxy\]
where $\pi_{1}$, $\pi_{2}$ are the natural projections
\[ \begin{array}{llllllllll}
 \pi_{1}\colon &  \F^3 \times \F^3 \times \F \times \F & \longrightarrow \F^3,& & & \pi_{2}: & \F^3 \times \F^3 \times \F \times \F  & \longrightarrow \F^3 \\
   & (\yy,\xx,\lambda,W) & \longmapsto \yy & & & & (\yy,\xx,\lambda,W) & \longmapsto \xx.
\end{array} \]
Then, the offset to $\ss$ at distance $d$ is defined as the Zariski closure of $\pi_{1}(\bb)$.

\subsection{Conchoidal construction}

Let $F(\xx)$ be the defining polynomial of an irreducible surface $\ss$, $\delta\in \F$ a non-zero element and $A\in \F^3$ (called the focus). Then, we consider the incidence variety
\[
 \bb=\left\{ (\yy,\xx,\lambda) \in \F^3 \times \F^3 \times \F  \left/
   \begin{array}{l}
         F(\xx)=0   \\
         \sum_{i=1}^{3}(x_{i}-y_{i})^2=\delta^2 \\
         \yy=A+\lambda (\xx-A)
   \end{array}
 \right.\right \}.
\]
The first equation ensures that $\xx \in \ss$, the second equation is the sphere center at $\xx$ with radius $d$, and the third equation is the line joining $\xx\in \ss$ with $A$. We consider the diagram
\[\xy
 (0,0)*+++{\bb \subset \F^3\times\F^3\times\F}="B";
 (-15,-20)*++{\pi_1(\bb)\subset\F^3}="k2";
 (15,-20)*++{\ss\subset\F^3}="k";
 {\ar_{\displaystyle \pi_1} "B"; "k2"};
 {\ar^{\displaystyle \pi_2} "B"; "k"};
 (55,-10)*{\mbox{(Conchoidal Incidence Diagram)}};
\endxy\]
where $\pi_{1}$, $\pi_{2}$ are the natural projections as above. Then, the conchoid to $\ss$ at distance $\delta$ from the focus $A$ is defined as the Zariski closure of $\pi_{1}(\bb)$.

\subsection{General geometric construction}

Now, in general, if $F(\xx)$ and $\ss$ are as above, we can define a \textsf{geometric construction} as an incidence diagram
\[\xy
 (0,0)*+++{\bb \subset \F^3\times\F^3\times\F^n\ (n\in {\mathbb N})}="B";
 (-15,-20)*++{\pi_1(\bb)\subset\F^3}="k2";
 (15,-20)*++{\ss\subset\F^3}="k";
 {\ar_{\displaystyle \pi_1} "B"; "k2"};
 {\ar^{\displaystyle \pi_2} "B"; "k"};
 (55,-10)*{\mbox{(General Incidence Diagram)}};
\endxy\]
where $\pi_{1}$, $\pi_{2}$ are the natural projections as above, and $\bb$ is the incidence variety
\[
 \bb=\left\{ (\yy,\xx,\zz) \in \F^3 \times \F^3 \times \F^n  \left/
 \begin{array}{l}
    F(\xx)=0   \\
    G_1(\xx,\yy,\zz)=0 \\ \,\,\,\,\vdots  \\
    G_m(\xx,\yy,\zz)=0
 \end{array}
  \right. \right \}.
\]
where $\zz=(z_1,\ldots,z_n)$ are the auxiliary variables and $G_1,\ldots,G_m$ are the polynomials defining the algebraic conditions of the geometric construction. In this situation, the Zariski closure of $\pi_1(\bb)$ is the \textsf{geometric variety} generated from $\ss$ via the geometric construction. In addition, we define the \textsf{degree of a geometric construction} as the degree of the map $\pi_2$; that is, as the cardinality of the anti-image via $\pi_2$ of a generic element in $\ss$.

We finish this section with the following theorem.

\begin{theorem}\label{theorem-geo-construction}
Let $\ss$ be a radical irreducible surface, and let $\mathcal{Z}$ be the geometric variety generated from $\ss$ via a geometric construction of degree at most 4. Then $\mathcal{Z}$ is radical.
\end{theorem}

\begin{proof}
Let the incidence variety $\bb$ of the geometric construction be defined by $\{F(\xx),G_1(\xx,\yy,\zz),\ldots,G_m(\xx,\yy,\zz)\}$. Let $\cP(\ot)$ be a radical parametrization of $\ss$, and let $\E$ be the last field in the tower root defining $\cP(\ot)$. Since the degree of the construction is at most $4$ we know that the algebraic system of equations $\{G_i(\cP(\ot),\yy,\zz)=0\}_{i=1,\ldots,m}$ has at most 4 solutions over the algebraic closure of $\E$. Let $\mathcal G$ be the reduced Gr\"obner basis of $\{G_i(\cP(\ot),\yy,\zz)\}_{i=1,\ldots,m}$ with respect to a lex order with $z_n>\cdots>z_1>y_n \cdots > y_1$; note that the ideal is considered in $\E[\yy,\zz]$.

Since the ideal is zero-dimensional, $\mathcal{G}$ is of the form
\[
 \{ g_{1,1}(\ot,y_1), g_{2,1}(\ot,y_1,y_2),\ldots,g_{2,k_2}(\ot,y_1,y_2),\ldots \}.
\]
Furthermore the degree of $g_{1,1}$ is at most 4, and every solution of $g_{1,1}$ over the algebraic closure of $\E$ can be continued to a solution of the full system (see e.g. \cite{win} p. 194). Therefore, solving by $g_{1,1}$, we express $y_1$ in terms of radicals. Now, for each root $\alpha$ of $g_{1,1}$, $\{g_{2,1}(\ot,\alpha,y_2)=0,\ldots,g_{2,k_2}(\ot,\alpha,y_2)=0\}$ has at most 4 roots. Therefore, $y_2$ can be expressed by radicals. So, by induction we get all $\yy$ expressed by radicals. Therefore, $\mathcal Z$ is parametrizable by radicals.
\end{proof}

Now, taking into account that the offset and conchoidal constructions are of degree 2, we have the following corollaries.

\begin{corollary}
The offset of a radical irreducible surface is radical.
\end{corollary}

\begin{corollary}
The conchoid of a radical irreducible surface is radical.
\end{corollary}

\section*{Acknowledgements}

The authors thank Josef Schicho and Michael Harrison for much of the background information provided on the cases of genus 1 and 4.


\begin{thebibliography}{}

\bibitem[Arbarello et al. (1985)]{ACGH}
Arbarello E., Cornalba M., Griffiths P.A., Harris J. (1985).
\textit{Geometry of algebraic curves. Vol. I}.
Grundlehren der Mathematischen Wissenschaften [Fundamental Principles of Mathematical Sciences], 267. Springer-Verlag, New York, 1985. xvi+386 pp. ISBN: 0-387-90997-4.

\bibitem[Arrondo et al. (1997)]{ASS}
Arrondo E., Sendra J., Sendra J.R. (1997).
\textit{Parametric Generalized Offsets to Hypersurfaces}.
J. Symbolic Comput. 23, no. 2-3, 267--285.

\bibitem[Arrondo et al. (1999)]{ASS2}
Arrondo E., Sendra J., Sendra J.R. (1999).
\textit{Genus Formula for Generalized Offset Curves}.
J. Pure Appl. Algebra 136, no. 3, 199--209.

\bibitem[Harrison (2011)]{Ha}
Harrison M. (2011).
\textit{Explicit Solution By Radicals, Gonal Maps and Plane Models of Algebraic Curves of Genus 5 or 6}.
To appear in the MEGA 2011 special issue of the J. Symbolic Comput., arXiv math.AG/1103.4946.

\bibitem[Hartshorne (1977)]{Hartshorne}
Hartshorne R. (1977).
\textit{Algebraic geometry}.
Graduate Texts in Mathematics, No. 52. Springer-Verlag, New York-Heidelberg, 1977. xvi+496 pp. ISBN: 0-387-90244-9.

\bibitem[Hess (2002)]{Hess}
Hess F. (2002).
\textit{Computing Riemann-Roch spaces in algebraic function fields and related topics}.
J. Symbolic Comput. 33, no. 4, 425--445.

\bibitem[Noether (1870)]{N70}
Noether M. (1870).
\textit{\"Uber Fl\"achen, welche Scharen rationaler Kurven besitzen. (German)}
Math. Ann. 3, no. 2, 161--227.

\bibitem[Peternell (1997)]{P97}
Peternell M. (1997).
\textit{Rational parametrizations for envelopes of quadric families}.
Ph.D. Thesis, Technical University of Vienna.

\bibitem[Peternell \& Pottmann (1998)]{PP}
Peternell M., Pottmann H. (1998).
\textit{A Laguerre Geometric Approach to Rational Offsets}.
Comput. Aided Geom. Design 15, no. 3, 223--249.

\bibitem[Peternell et al. (2011)]{PGS}
Peternell M., Gruber D., Sendra J. (2011).
\textit{Conchoid surfaces of rational ruled surfaces}.
Comput. Aided Geom. Design 28, no. 7, 427--435.

\bibitem[Saint-Donat (1973)]{Saintdonat}
Saint-Donat B. (1973).
\textit{On Petri's analysis of the linear system of quadrics through a canonical curve}.
Math. Ann. 206, 157--175.

\bibitem[Schicho (1998a)]{Sch-a}
Schicho J. (1998a).
\textit{Rational parametrization of surfaces}.
J. Symbolic Comput. 26, no. 1, 1--29.

\bibitem[Schicho (1998b)]{Sch-b}
Schicho J. (1998b).
\textit{Rational parameterization of real algebraic surfaces}.
Proceedings of the 1998 International Symposium on Symbolic and Algebraic Computation (Rostock), 302--308 (electronic), ACM, New York, 1998.

\bibitem[Sendra \& Sendra (2010)]{SeSe}
Sendra J., Sendra J.R. (2010).
\textit{Rational Parametrization of Conchoids to Algebraic Curves}.
Appl. Algebra Engrg. Comm. Comput. 21, no. 4, 285--308.

\bibitem[Sendra \& Sevilla (2011)]{SeSevilla}
Sendra J.R., Sevilla D. (2011).
\textit{Radical Parametrizations of Algebraic Curves by Adjoint Curves}.
J. Symbolic Comput. 46, no. 9, 1030--1038.

\bibitem[Sendra et al. (2008)]{libro}
Sendra J.R., Winkler F., P{\'e}rez-D{\'{\i}}az S. (2008).
\textit{Rational algebraic curves: A computer algebra approach}.
Algorithms and Computation in Mathematics, 22. Springer, Berlin, 2008. x+267 pp. ISBN: 978-3-540-73724-7.

\bibitem[Walker (1978)]{walker}
Walker R. (1978).
\textit{Algebraic curves}.
Reprint of the 1950 edition. Springer-Verlag, New York-Heidelberg, 1978. x+201 pp. ISBN: 0-387-90361-5.

\bibitem[Winkler (1996)]{win}
Winkler F. (1996).
\textit{Polynomials Algorithms in Computer Algebra}.
Texts and Monographs in Symbolic Computation. Springer-Verlag, Vienna, 1996. viii+270 pp. ISBN: 3-211-82759-5.

\bibitem[Zariski (1926)]{Zariski1926a}
Zariski O. (1926).
\textit{Sull'impossibilit\`{a} di risolvere parametricamente per radicali un'equazione algebrica $f(x,y)=0$ di genere $p>6$ a moduli generali}.
Atti Accad. Naz. Lincei Rend., Cl. Sc. fis. Mat. Natur., serie VI vol 3, 660--666.

\end{thebibliography}
\end{document}